\documentclass[12pt,letter]{article}
\usepackage[T1]{fontenc}
\usepackage{lmodern}
\usepackage[latin1]{inputenc}
\usepackage{setspace}
\usepackage{etoolbox}
\makeatletter
\patchcmd{\appendix}{\@Alph}{\@Roman}{}{}
\makeatother

\usepackage{amsmath}
\usepackage{amssymb}
\usepackage{amsthm}
\usepackage{mathtools}
\usepackage{mathrsfs}
\usepackage{breqn}

\usepackage[margin=1 in]{geometry}
\usepackage{enumerate}
\usepackage{enumitem}
\setlist[enumerate,1]{label=(\arabic*)}
\setlist[itemize,1]{label=--}    
\usepackage[dvipsnames]{xcolor}
\usepackage{hyperref}
\usepackage{float}
\usepackage{indentfirst}

\usepackage{tikz}
\usetikzlibrary{decorations.pathreplacing}
\usepackage{mathtools}
\usetikzlibrary{positioning}
\usepackage{graphicx}

\newcommand{\bb}{\mathbb}

\newcommand{\und}{\underline}
\newcommand{\lims}{\lim\limits}

\newcommand{\notext}{\noindent \textbf}

\newcommand{\mcal}{\mathcal}

\renewcommand{\epsilon}{\varepsilon}

\newcommand{\blue}[1]{\color{blue}#1 \color{black}}

\newtheorem{theorem}{Theorem}
\newtheorem{lemma}{Lemma}
\newtheorem{definition}{Definition}
\newtheorem{proposition}{Proposition}
\newtheorem{corollary}{Corollary}

\newtheoremstyle{remboldstyle}
  {}{}{}{}{\bfseries}{.}{.5em}{{\thmname{#1 }}{\thmnumber{#2}}{\thmnote{ (#3)}}}
\theoremstyle{remboldstyle}
\newtheorem{rembold}{Example}

\DeclareTextFontCommand{\emph}{\slshape}

\usepackage{natbib}
\nocite{*}

\title{When and Where To Submit A Paper}
\author{Daniel Luo\footnote{Department of Economics, Massachusetts Institute of Technology. daniel57@mit.edu.
\\ I am particularly indebted to Ryo Shirakawa for many enlightening discussions regarding this project. I am also grateful to Laura Doval, Drew Fudenberg, Eric Gao, Yucheng Shang, Vitalii Tubdenov, and Myles Winkley for feedback and suggestions. I acknowledge financial support from the NSF Graduate Research Fellowship Program. All errors are, of course, my own.}}
\date{\today}
\begin{document}

\maketitle
\begin{abstract}
What is the optimal order in which a researcher should submit their papers to journals of differing quality?
I analyze a sequential search model without recall where the researcher's expected value from journal submission depends on the history of past submissions. 
Acceptances immediately terminate the search process and deliver some payoff, while rejections carry information about the paper's quality, affecting the researcher's belief in acceptance probability over future journals.
When journal feedback does not change the paper's quality, the researcher's optimal strategy is monotone in their acceptance payoff. 
Submission costs distort the researcher's effective acceptance payoff, but maintain monotone optimality. If journals give feedback which can affect the paper's quality, such as through referee reports, the search order can change drastically depending on the agent's prior belief about their paper's quality. However, I identify a set of \textit{assortative matched} conditions on feedback such that monotone strategies remain optimal whenever the agent's prior is sufficiently optimistic. 
\end{abstract}

\notext{Keywords}: Sequential Search. Restless Bandits. No-Recall. Paper Submission. 

\noindent \bigskip \textbf{JEL Codes}: C44, D81, D83. 

\newpage 
\onehalfspacing 
\section{Introduction}
Consider the problem faced by a researcher choosing between different journals to submit their manuscript to.
The researcher is uncertain about the quality of their paper, and hence their probability of acceptance from the journals. 
Journals are differentiated along three dimensions: they accept papers at different rates, provide referee reports of various quality, and carry varying value to the researcher if the paper is accepted. 
Submitting to a journal thus carries two benefits: if accepted, the researcher can secure their acceptance payoff and end (potentially costly) search, while if rejected they can receive information about the quality of their paper, which can lead to revisions that improve the manuscript. 

These two forces may sometimes act in opposite directions and thus require the researcher to make trade-offs when evaluating which journal to submit to at any point in time.
For example, if the researcher submits their paper to a less-preferred journal, they can incorporate the feedback produced in the referee reports and improve their paper's quality. This may then increases the future probability they are accepted to a more-preferred journal. 
However, eliciting feedback in this manner can be risky if it is infeasible for the researcher to retract their article from consideration if they are accepted, since submitting to lower-ranked journals may foreclose the possibility of submitting to higher-ranked journals in the future. 
The researcher thus needs to decide between the potential for acceptance \textit{today} with the rich informational dynamics that are encoded in a rejection, which can affect their probability for future acceptance (and thus potentially higher payoff) tomorrow. 

To study the interaction between these economic forces, I study a simple model of sequential search without recall where the value of search is correlated across objects and history-dependent. An agent who is uncertain about the quality of their paper decides the order in which they submit to journals, which can either accept or reject the agent's paper.
Journals accept high-quality papers with some probability, which end the search process for the agent and nets them some flow payoff, while they always reject low-quality papers.
Rejections are not materially beneficial to the researcher, and eliminate the ability for future resubmission. However, the rejection itself may carry information about the paper's quality, and also carries with it referee reports which may increase the paper's quality, and thus may affect the researcher's perceived value from submitting to each of the remaining journals.  

A natural benchmark strategy to consider is the one where the researcher submits to journals in order of their value (if accepted), regardless of the informational content carried by the rejection. Theorem \ref{no_feedback} shows that when journals do not give useful referee reports, this \textit{monotone} strategy is optimal: the informational content conveyed by rejections is always second order, regardless of the agent's prior belief about their paper's quality. When referee reports are useful, I show that monotone strategies remain optimal across all prior beliefs if and only if the researcher's continuation value depends only on the set of journals a researcher has already submitted to, and not the order in which they submitted to those journals.
This \textit{order-independence} condition holds only when journal acceptance rates and feedback strengths satisfy a constant marginal benefit condition, where the ratio of acceptance rates is the inverse ratio of the feedback strengths across any two journals. 

When there are more than two distinct journals, order-independence generically fails to be satisfied, and thus the problem becomes much more complicated. Since the value of feedback depends on the researcher's belief about their paper's quality, no \textit{prior-free} index exists: the optimal submission order can depend crucially on the researcher's prior belief about their paper's quality. 
However, so long as the researcher's prior belief about their paper's quality is sufficiently high, Theorem \ref{weak_feedback} proves monotone strategies remain optimal when journal characteristics satisfy an \textit{assortative matching} condition. In particular, I require that more valuable journals have lower acceptance rates and give stronger feedback, and that each journal's value is sufficiently distinct. 
Finally, I show these assortative matching conditions are minimally sufficient for monotone optimality: removing any of the (ordinal) conditions falsifies the theorem, even while imposing the other conditions. 

Methodologically, the paper identifies conditions under which locally monotone \textit{pairwise switches} are profitable. The conditions in Theorem \ref{weak_feedback} provide a novel bound on the total possible change in continuation value from \textit{all future journal submissions} if the paper is rejected from both journals which are involved in the switch that may be useful more generally in dealing with correlated bandit problems. 

There are several other settings where the core economic frictions of this paper may come in conflict with one another. 
For example, an employee who is interviewing for firms may apply to lower-pay firms in order to learn more about their ability and practice their technical interview skills. High school students may apply to a variety of schools in the early action round to gauge the competitiveness of their profiles before the regular decision round; firms may choose to cooperate on a smaller project to see whether or not they can form a productive relationship later. In all of these cases, the agent needs to balance the instrumental value from a success today with learning about value tomorrow, exactly the core tradeoff encapsulated in the model. 

The rest of this paper is organized as follows. The remainder of this section reviews the related literature. \blue{Section \ref{model_section}} introduces the formal model and walks through some examples. \blue{Section \ref{results_section}} analyzes the model and discusses important qualitative results. \blue{Section \ref{conclusion_section}} concludes. The \blue{\ref{Appendix}} collects proofs not included in the main body, particularly computational lemmas, and several auxiliary results and examples.

\subsection{Related Literature}

The problem of optimally searching between one of many potential alternatives was first formulated and solved by \cite{weitzman1979optimal}, who constructed a time-invariant \textit{index} which governed the optimal search history. Further variants of the problem, which vary the rewards or the costs needed to open a box, have been studied by \cite{olszewski2015general}, \cite{doval2018pandoras}, and \cite{chade2006simultaneous}.
\cite{vishwanath1992parallel} and \cite{morgan1985optimal} generalize these search models and consider optimal strategies when the agent can open multiple boxes as once. 
\cite{zhou-2011} studies sequential search when boxes are differentiated in two different dimensions. These papers offer a litany of motivations for the problem, including job search, college admissions, consumer choice, lab experimentation, highlighting the ubiquity of the underlying economic structure of the problem. Importantly, the value of each box is independent in these situations; Theorem \ref{no_feedback} contributes to this literature by identifying a simple optimal index under a parametricized form of common correlation. 

Methodologically, these search models relate to the multi-armed bandit problem, first formulated by \cite{wald1947foundations} and solved by \cite{gittins1979bandit} under the assumption that only the sampled arm in any period underwent a state transition. 
\cite{whittle1988restless} introduced the notion of a \textit{restless bandit}, where the arms which are not pulled can also undergo a state transition. However, \cite{papadimitriou1999complexity} showed that in general, computing the optimal index is computationally hard, and as a result the restless bandit problem is computationally difficult. Partial progress towards this problem in some cases is further reviewed by \cite{ninomora2023markovian}. This paper relates to this literature by introducing a specific form of \textit{common-valued} correlation and identifying the relevant index in this setting under some parametric assumptions. 

Finally, effect of \textit{information provision} on the optimal search policy in theses settings is also well studied. \cite{ke2019optimal} allows the agent to pay for information before they begin to engage in search, while \cite{sato2023information} and \cite{sato2023persuasion} considers a designer who chooses a Blackwell experiment with the implicit goal of maximizing the length of search. \cite{au2023attraction}, \cite{boardLuSearchMarkets} and \cite{he-li-2023}, look at the effect of competitive information disclosure under different assumptions on the consumer's search behavior. Finally, \cite{che2019optimal} considers optimal search when the results of search affect the prior and thus the value of the \textit{entire problem}, which is most similar to the structure of this paper. This paper contributes to this literature by analyzing a model of sequential search with \textit{exogenous} information provision but also \textit{exogenous} search. In addition, I focus on the optimal search index instead of information provision, which is central in many of these papers (in particular \cite{au2023attraction}). 

\section{Model}
\label{model_section}
The formal model is as follows. A single agent chooses between finitely many journals, indexed by $\{1, 2, \dots, I\}$, for some $I \in \bb{N}$. 
Their paper can either be \textit{high quality} (H) or \textit{low quality} (L), and the research has a prior belief $\mu \in \Delta(\{H, L\})$ over the quality of their paper.
Each journal is parametricized by four characteristics: a payoff $u_i > 0$ conditional on acceptance, a probability $a_i$ with which they accept high quality papers (low quality papers are never accepted, by any journal), a submission cost\footnote{We interpret this submission cost flexibly: it can be a literal submission cost in monetary units to a journal, or it can capture psychological costs imposed by an expected rejection, the costs of delay imposed by slow turnaround times for a journal, etc.} $c_i$, and a \textit{feedback probability} $q_i$ which with low-quality papers become high quality papers. I will refer to $\{u_i, a_i, c_i, q_i\}_{i = 1}^I$ as the ambient environment for the agent facing $I$ boxes. Moreover, I will reindex the boxes in decreasing order of their acceptance payoffs, so that $u_i \geq u_j$ whenever $i \geq j$. Finally, I make a genericity assumption that each journal has a distinct acceptance payoff; this will not affect qualitative results but will simplify the exposition by removing the need to carefully appeal to tiebreaking rules. 

Once an agent has submitted to a journal, they can no longer submit to any more journals. As a result, the agent will exhaust all of their options with positive probability and have no journals left to submit to. If this is the case, I assume the agent takes an outside option, which is normalized to have value $0$. Normalizing the outside option to $0$ is without loss of optimality (as shown in \blue{\ref{Appendix III}\textcolor{black})}but simplifies the analysis. 

The timing of the game is as follows. 
In each period $t$, the agent has a belief $\mu_t$ about the quality of their paper and a set of journals still available, $I_t \subset \mcal I$. Given $(\mu_t, I_t)$, the agent chooses a journal $i \in I_t$ to submit to and pay cost $c_i$ for sure. With probability $a_i \mu_t(H)$, they are accepted by journal $i$, receive payoff $u_i$, and exit the game. With complementary probability $1 - a_i \mu_t(H)$, they set $I_{t + 1} = I_t \setminus \{i\}$. If $I_{t + 1} = \varnothing$, the game ends and the agent collects their outside option. Otherwise, they proceed into period $t + 1$ with available set of journals $I_{t + 1}$ and belief 
\[ \mu_{t + 1}^i = \begin{bmatrix} 1 & q_i \\ 0 & 1 - q_i \end{bmatrix} \begin{bmatrix} \frac{(1 - a_i) \mu_t(H)}{1 - a_i \mu_t(H)} \\ \frac{1 - \mu_t(H)}{1 - a_i \mu_t(H)} \end{bmatrix} \]
where the square matrix represents the transition probabilities over states that are a result of the \textit{feedback} channel of the journal and the column represents the agent's belief about their paper's quality conditional on rejection but when feedback is $0$. 

The posterior belief, conditional on rejection, depends on the acceptance and feedback rates $(a_i, q_i)$ and the prior belief $\mu_t$. As a piece of notational shorthand, let $f_i: \Delta(\{H, L\}) \to \Delta(\{H, L\})$ be the map which performs the above computation for a box $i$ given belief $\mu_t$. 
\[ f_i(\mu_t)(H) = \frac{(1 - a_i - q_i)\mu_t(H) + q_i}{1 - a_i \mu_t(H)} \]
At every time $t$, the state is given by $(\mu_t, I_t)$. Because rejections only affect an agent's utility at time $t$ by restricting the set of available journals and changing their posterior belief, the state captures the entire history of past submissions. A strategy is then a function 
\[ \sigma: \Delta(\{H, L\}) \times 2^{\mcal I} \to \mcal I \text{  such that  } \sigma(\mu_t, I_t) \in I_t\]
for all subsets $I_t$ of $\mcal I$. This is a little bit of a complicated object to work with. However, since rejection ends the game and there is no recall, it is possible to compute the \textit{unique belief} that arises in continuation play given strategy $\sigma$ and state $(\mu_t, I_t)$ as $f_{\sigma(\mu_t, I_t)}(\mu_t)$.

Thus, given a prior belief $\mu_0$ and a strategy $\sigma$, define a sequence of induced beliefs $\{\mu_t^\sigma\}$ and an induced \textit{permutation} $\tau_\sigma \in \mcal S_{I}$ over boxes inductively as follows: 
\begin{enumerate}
    \item Fix a prior $\mu_0$. Define $\tau_\sigma(1) = \sigma(\mu_0, \mcal I)$ and $\mu_1^\sigma = f_{\tau_\sigma(1)}(\mu_0)$. 
    \item For any $t \leq I$, and the induced sequence $\{\mu_s^\sigma, \tau_\sigma(s)\}_{s = 1}^{t - 1}$, define $\tau_\sigma(t) = \sigma(\mu_{t - 1}^{\sigma}, I_t)$ where $I_t = \mcal I \setminus \{\sigma(s)\}_{s = 1}^{t - 1}$. Similarly, inductively define the belief $\mu_t^\sigma = f_{\tau_\sigma(t)}$.
    \item If $t < I$, continue with step (2) again with the sequence $\{\mu_s^\sigma, \tau_\sigma(s)\}_{s = 1}^{t}$ and time $t + 1$. If $t = I$, define the \textit{observed inspection order} given $\sigma$ to be
    \[ \{\mu_s^\sigma, \tau_\sigma(s)\}_{s = 1}^{I} \]
\end{enumerate}
I will sometimes refer to $\tau$ as a \textit{search order} and $\{\mu_s^\sigma\}$ as the \textit{induced beliefs.}
Note that each $\sigma$ induces a unique inspection order. In particular, $\tau_\sigma$ is deterministic, since acceptances end the search process. 
Each strategy $\sigma$ also induces a probability distribution $p_\sigma(t)$ over the probability that an agent will have to search over at least $t$ boxes, which is defined by the function 
\[ p_\sigma(t) = \prod_{s = 1}^t(1 - a_{\tau_\sigma(s)} \mu_s^\sigma(H)) \text{   where  }  p_\sigma(I + 1) = 1 - p_\sigma(I) \]
is the probability that the agent is rejected from \textit{every journal} (and thus cannot submit their paper). 
Given the induced probability distribution $p_\sigma$, the agent's expected payoff from their time $0$ perspective given a strategy $\sigma$ is given by 
\[ \mcal U(\sigma) = \sum_{t = 1}^{I} p_\sigma(t) \left( u_{\tau_\sigma(t)} a_{\tau_\sigma(t)} \mu_t(H) - c_{\tau_\sigma(t)} \right) \]

The goal is to understand which strategies are \textit{optimal} in the sense that they net the designer the highest possible payoff. Formally,

\begin{definition}
A strategy $\sigma$ is \und{optimal} if, for every other strategy $\sigma'$, $\mcal U(\sigma) \geq \mcal U(\sigma')$. 
\end{definition}

Strategies are relatively complicated objects, since they need to specify actions for \textit{every possible} posterior belief at every possible set of available journals. In practice, however, the induced permutation $\tau_\sigma$ uniquely pins down the sequence of on-path posterior beliefs $\{\mu_t^\sigma\}$, and hence determines the entire path of play. As a result, payoffs between two strategies should only differ if they induce different permutations; our first result shows it is sufficient to work with this (much smaller) class of \textit{induced permutations} to verify optimality. 

\begin{lemma}[Coherency Lemma]
\label{coherency}
    For each permutation $\tau \in \mcal S_{I}$, there exists a strategy $\sigma$ such that $\tau$ is induced by $\sigma$. Moreover, if $\tau_\sigma = \tau_{\sigma'}$, then $\mcal U(\sigma) = \mcal U(\sigma')$.
\end{lemma}
\begin{proof}
    First, fix $\tau$. Define $\sigma$ inductively first with $\sigma(\mu, \mcal I) = \tau(1)$ and then with $\sigma(\mu, I_t) = \tau(t)$ for every $\mu$. This guarantees $\tau_\sigma = \tau$, with the belief path $\{\mu_t^\sigma\}$ pinned down by the prior and $\tau$. Similarly, note that if $\tau_\sigma = \tau_{\sigma'}$, then 
   \begin{align*}
       \mcal U(\sigma) = & \sum_{t = 1}^{I} \left(\prod_{s = 1}^t(1 - a_{\tau_\sigma(s)} \mu_s^\sigma(H))\right)\left( u_{\tau_\sigma(t)} a_{\tau_\sigma(t)} \mu_t^\sigma(H) - c_{\tau_\sigma(t)} \right)
       \\ = & \sum_{t = 1}^{I} \left(\prod_{s = 1}^t(1 - a_{\tau_{\sigma'}(s)} \mu_s^{\sigma'}(H))\right)\left( u_{\tau_{\sigma'}(t)} a_{\tau_{\sigma'}(t)} \mu_t^{\sigma'}(H) - c_{\tau_{\sigma'}(t)} \right)= \mcal U(\sigma')
   \end{align*}
   so these two strategies give the same payoffs. 
\end{proof}

The coherency lemma is obvious but simplifies the problem of finding an optimal strategy into that of finding an optimal \textit{search order}, given the prior belief. Because the induced beliefs are uniquely pinned down by the beliefs and the search order $\tau_\sigma$ given a strategy $\sigma$, optimality of any strategy will depend only on (1) the payoff characteristics of the boxes, (2) the induced search order $\tau_\sigma$, and (3) the prior belief. 
The goal is to understand, for which prior beliefs and which sets of box characteristics, the optimal strategy is to submit to journals with the highest acceptance payoffs. Formally, strategies with this property will be known as \textit{monotone} strategies.

\begin{definition}
\label{monotonicity_defn}
    A strategy $\sigma$ is \und{monotone} if its induced permutation $\tau_\sigma$ is the identity. 
\end{definition}

That Definition \ref{monotonicity_defn} implies the characterization above follows from the fact that journals are indexed in decreasing order of their acceptance payoffs $u_i$. When there are no submission costs, this is the behavior one would expect from an agent who behaves \textit{as-if} they are myopically optimizing, without considering the information carried in a rejection. 

There are two potential reasons why the agent's optimal strategy may not be monotone. The first is that submission costs may ``reverse'' the \textit{true} payoff that an agent faces: if two journals have relatively similar acceptance payoffs but drastically different submission costs, the latter will loom large in the agent's decision. The second problem is that the feedback channel \textit{may} be important: if a journal with a slightly lower payoff gives much better feedback, then a researcher who is aware of a potentially low likelihood of being accepted to the first journal they submit to may choose to submit first to the journal with a weaker acceptance payoff but stronger feedback then vice versa. Below is a stark example that highlights this friction. 

\begin{rembold}
    \label{strong_feedback}
Suppose that $I = 2$. Suppose each journals has the following characteristics:
\begin{table}[h]
\centering
\begin{tabular}{|c|c|c|c|c|}
\hline
          & $u_i$ & $a_i$ & $q_i$ & $c_i$ \\ \hline
Journal 1 & 5     & 0.2   & 0.2   & 0     \\ \hline
Journal 2 & 1     & 0.3   & 0.4   & 0     \\ \hline
\end{tabular}
\end{table}

\noindent Because there are only two boxes, there are only two possible strategies, and the induced paths of beliefs are straightforward to compute. We know that monotone strategies are optimal if and only if 
\[ u_1 a_1 \mu + u_2 a_2 (1 - a_1 \mu) f_1(\mu) \geq u_2 a_2 \mu + u_1 a_1(1 - a_2\mu) f_2(\mu)\]
Some algebraic simplification implies this is equivalent to 
\[ (u_2 a_2 q_1 - u_1 a_1 q_2)(1 - \mu) \geq a_1 a_2(u_2 - u_1) \mu \iff \frac{17}{12} \leq \frac{\mu(H)}{\mu(L)} \iff \mu(H) \geq \frac{17}{29} \]
Why are nonmonotone strategies optimal in this case? First let $\mu(H) \leq \frac47 < \frac{17}{29}$. Then 
\[ f_2(\mu)(H) = (1 - a_2 \mu(H)) \frac{(1 - a_2 - q_2)\mu(H) + q_2}{(1 - a_2 \mu(H))} = 0.3 \mu(H) + 0.4 \geq \mu(H)\]
that is, the agent's posterior belief about the quality of their paper \textit{increases} after a rejection, since the increase in the paper's expected quality after reading reports is greater than the informational content of a rejection. 
We will refer to box characteristics and beliefs where the posterior belief of a paper's quality is greater after a rejection as a \textit{strong feedback} situation. 
However, there also exist prior beliefs $\mu \in \left(\frac47, \frac{17}{29}\right)$ where the nonmonotone strategy is optimal but primitives do not feature strong feedback. In this case, optimality of nonmonotinicity comes from the fact that $f_2(\mu) \approx \mu$, but $f_1(\mu) << \mu$, and hence the agent's desire to ``smooth'' their posterior belief over nodes in the game outweighs the desire to frontload payoffs by submitting to journal $1$ first. 
\end{rembold}

Example \ref{strong_feedback} both demonstrates that the problem of finding an optimal, prior-free index may generally not be possible, and that feedback can greatly affect the agent's optimal strategy. To identify the direct effect of submission costs and feedback on optimal strategies, I first consider the model without feedback and identify a prior-free index for which monotone strategies are optimal. Second, I consider the model without submission costs (so payoffs are independent of acceptance rates) and show that, under suitable \textit{joint} regularity conditions on payoffs and posterior beliefs, the naively monotone strategy is optimal. 

\section{Conditions for Monotonicity}
\label{results_section}
\subsection{The No-Feedback Case}
First, consider the case where journals do not provide referee reports, e.g. $q_i = 0$ for all $i$. Here, journals carry information about the paper's quality \textit{only through} their acceptance and rejection decisions. 
This implies two special properties which dramatically simplify the correlational structure between journals and which will allow for a clean indexability result. First, beliefs must drift \textit{downward}: that is, rejections make the agent more pessimistic about the quality of their paper, regardless of what their prior beliefs are, since \textit{only} high-type papers are accepted with some probability. 
Second, the agent's posterior belief at any point depends only on the sequence of journals that they have submitted to, and not the \textit{order} with which they submit to these journals. To see why, note that without the state transitions that occur when $q_i > 0$, each journal can be viewed as a Blackwell experiment, where the signals are $\{\text{Accept, Reject}\}$. Because a posterior belief given a sequence of blackwell experiments are independent of the order with which the signals associated to these experiments are received, the agent's posterior belief (conditional on staying in the game) depends only on which journals at any point in time have rejected their paper. 

Both of these properties need not hold when feedback is nonzero. 
First, Example \ref{strong_feedback} showed that, for $\mu(H) < \frac47$, $f_2(\mu)(H) > \mu(H)$. Second, conditional on being rejected from both journals, some algebra gives that 
\[ f_2(f_1(\mu))(H) = \frac{0.18\mu + 0.46}{(1 - 0.2\mu)(1 - 0.3f_1(\mu))} \text{   and  } f_1(f_2(\mu))(H) = \frac{0.18\mu + 0.44}{(1 - 0.2f_2(\mu))(1 - 0.3\mu)}  \]
which are generically not equivalent. Thus, the order of submission affects the posterior belief conditional on rejection.
This second property is particularly important, because it greatly simplifies an agent's continuation belief $\mu_i^\sigma$ and disciplines the degree to which payoffs can change \textit{later} in the sequence due to changes in the order of submission earlier in the strategy. This property will be crucial to identifying the relevant index. Lemma \ref{order_indep_lemma} gives a complete characterization in terms of model primitives for when this property holds. 

\begin{definition}
    Boxes $\{u_i, a_i, q_i, c_i\}$ is \textit{order-independent} if, for any strategies $\sigma$ and $\sigma'$ and any $t$ such that $\{\tau_\sigma(s)\}_{s = 1}^t = \{\tau_{\sigma'}(s)\}_{s = 1}^t$, then $\mu_t^{\sigma} = \mu_t^{\sigma'}$. 
\end{definition}

\begin{lemma}
\label{order_indep_lemma}
    Fix any two boxes with acceptance rates $(a_1, a_2)$ and feedback rates $(q_1, q_2)$, and full-support prior $\mu \in \Delta(\{H, L\})$. Then $f_1(f_2(\mu)) \geq f_2(f_1(\mu))$ if and only if $a_1 q_2 \geq a_2 q_1$, with equality if and only if $a_2 q_1 = a_1 q_2$. 
\end{lemma}

A formal proof of Lemma \ref{order_indep_lemma} is given in \blue{\ref{order_indep_lemma_proof}}. For a brief mathematical intuition, note the cross terms $a_i q_j$ and $q_j a_i$ are the only terms which can appear in $f_j(f_i(\mu))$ and $f_i(f_j(\mu))$ which are asymmetric in $i$ and $j$. So long as these two values are the same, the closed-form expression for the posterior belief will be symmetric, and hence be order-independent. 
In addition, one interpretation of the ratio $\frac{q_i}{a_i}$ is as the relative strength of the feedback and acceptance effects on the agent's belief after submitting to journal $i$. If $\frac{q_2}{a_2} > \frac{q_1}{a_1}$, then the posterior belief is larger after submitting to the second journal (relative to the first one), so the agent maximizes their posterior belief by submitting to journal $2$ first. 


With these properties in mind, we can now construct an index, which does not depend on the agent's prior belief, which uniquely identifies the optimal strategy. The idea behind the proof is to find a strategy $\sigma$ which does not respect this index and then construct a profitable ``pairwise swap,'' using pairwise independence and the facts that the \textit{total probability of exit} do not depend on the order of submission to discipline the effect that swaps have on continuation play. The formal argument (and the index) are given below. 

\begin{theorem}[No-Feedback Indexability]
\label{no_feedback}
    Let $q_i = 0$ for all $i$. The optimal strategy induces a permutation $\tau_{\sigma}$ such that the \textit{modified acceptance payoffs} $u_{\tau_{\sigma}(t)} - \frac{c_{\tau_{\sigma}(t)}}{a_{\tau_{\sigma}}(t)}$ are 
    decreasing in $t$. 
\end{theorem}
\begin{proof}
The argument is by contradiction. Let $\tau_\sigma$ be some permutation which is nonmonotone in this index. Nonmonotonicity implies that there exists some $i$ such that $\tau_\sigma(i) > \tau_\sigma(i + 1)$. Let $\tau$ be a permutation that $\tau$ switches the $i$ and $i + 1$-th boxes which are inspected, while fixing the rest of the inspection order. By Lemma \ref{coherency}, there exists a strategy $\sigma'$ which induces $\tau$. The payoff differential between $\sigma$ and $\sigma'$ is given by 
\begin{align*}
    \mcal U(\sigma) - \mcal U(\sigma') = & \sum_{t = 1}^I \left(p_{\sigma}(t)\left( u_{\tau_\sigma(t)} a_{\tau_\sigma(t)} \mu_t^\sigma(H) - c_{\tau_\sigma(t)} \right) - p_{\sigma'}(t)\left( u_{\tau_{\sigma'}(t)} a_{\tau_{\sigma'}(t)} \mu_t^{\sigma'}(H) -  c_{\tau_{\sigma'}(t)}\right) \right)
    \\ = & \sum_{t = i}^{i + 1} \left(p_{\sigma}(t)\left( u_{\tau_\sigma(t)} a_{\tau_\sigma(t)} \mu_t^\sigma(H) - c_{\tau_\sigma(t)} \right) - p_{\sigma'}(t)\left( u_{\tau_{\sigma'}(t)} a_{\tau_{\sigma'}(t)} \mu_t^{\sigma'}(H) -  c_{\tau_{\sigma'}(t)}\right) \right)
    \\ + & \sum_{t = i + 2}^I \left( p_{\sigma}(t) - p_{\sigma'}(t) \right)\left( u_{\tau_\sigma(t)} a_{\tau_\sigma(t)} \mu_t^\sigma(H) - c_{\tau_\sigma(t)} \right) 
\end{align*}
where we use the fact that the inspection orders for $\sigma$ and $\sigma'$ are the same outside of $\{i, i + 1\}$, and that $\mu_t^{\sigma} = \mu_t^{\sigma'}$ for all $t > i + 1$ by order-independence. 
Similarly, it must be that $p_\sigma(i) = p_{\sigma'}(i + 1)$. Thus, the payoff differential on the second line, which is the difference in payoffs between $\sigma$ and $\sigma'$ at terminal nodes where either journals $\tau_{\sigma}(i)$ or $\tau_{\sigma}(i + 1)$ accept the paper, is given by 
\begin{align*}
    p_\sigma(i)\big( u_j a_j \mu_i^\sigma(H) - c_j + (1 - a_j \mu_i^\sigma(H))(u_k a_k f_j(\mu_i^\sigma)(H) - c_k)
    \\ -  u_k a_k \mu_i^\sigma(H) + c_k - (1 - a_k \mu_i^\sigma(H))(u_j a_j f_k(\mu_i^\sigma)(H) - c_j)\big)
\end{align*}
where we set $j = \tau_{\sigma}(i)$ and $k = \tau_{\sigma}(i + 1)$ for notational brevity. Dividing through by the positive $p_\sigma(i)$ and rearranging terms implies this difference is negative if and only if 
\begin{align*}
& \quad u_j a_j \mu_i^\sigma(H) - c_j + u_k a_k (1 - a_j) \mu_i^\sigma(H) - (1 - a_j \mu_i^\sigma(H)) c_k 
\\ \leq & \quad u_k a_k \mu_i^\sigma(H) - c_k + u_j a_j(1 - a_k)\mu_i^\sigma(H) - (1 - a_k \mu_i^\sigma(H))c_j
\\ \iff & \quad u_j a_j a_k \mu_i^\sigma(H) - u_k a_k a_j \mu_i^\sigma(H) + a_j c_k \mu_i^\sigma(H) - a_k c_j \mu_i^\sigma(H) \leq 0 
\\ \iff & \quad a_j a_k \mu_i^\sigma(H) \left(\left( u_j - \frac{c_j}{a_j}\right) - \left( u_k - \frac{c_k}{a_k} \right)\right) \leq 0
\end{align*}                                                                        The first equivalence follows by cancelling symmetric terms; the second equivalence is true by our definition of monotonicity since $k < j$. 
If $i + 1 = I$, this implies $\mcal U(\sigma) < \mcal U(\sigma')$, contradicting optimality of $\sigma$ and finishing the proof. Otherwise, note that 
\begin{align*}
p_\sigma(i + 2) = 1 - \sum_{t = 0}^{i + 1} a_{\sigma(t)} \mu_t^\sigma(H) = 1 - \sum_{t = 0}^{i - 1} a_{\sigma(t)}\mu_t^{\sigma}(H) - a_j \mu_i^{\sigma}(H) - a_k(1 - a_j)\mu_i^{\sigma}(H)
\\ =  1 - \sum_{t = 0}^{i - 1} a_{\sigma'(t)}\mu_t^{\sigma'}(H) - a_k \mu_i^{\sigma'}(H) - a_j(1 - a_k)\mu_i^{\sigma'}(H) = p_{\sigma'}(i + 2)
\end{align*}    
This computation, along with order independence, implies $p_\sigma(t) = p_{\sigma'}(t)$ for all $t > i + 1$. Thus, the terms on the third line in the payoff differential are $0$, and hence $\mcal U(\sigma) < \mcal U(\sigma')$. 
\end{proof}

\begin{corollary}
    If $c_i = 0$ for all $i$, then a strategy $\sigma^*$ is optimal if and only if it is monotone. 
\end{corollary}

The proof of this statement follows immediately from Theorem \ref{no_feedback} and the fact that in this case, the optimal index is $\{u_i\}$. 
The two key properties behind this proof -- order independence and invariance of the total probability of exit -- can be generalized to settings where journals give feedback, so long as the condition $a_i q_j$ is constant across all $(i, j)$ in Lemma \ref{order_indep_lemma} is satisfied.
However, note that away from the no-feedback case, the acceptance rate $a_i$ affects both the strength of feedback and the relative \textit{expected} payoff differential of the index identified in Theorem \ref{no_feedback}. For ease of exposition and simplicity in constructing the optimal index, we will assume submission costs are sufficiently small that the index in Theorem \ref{no_feedback} is exactly the index where ex-post payoffs $u_i$ is decreasing. 

\begin{proposition}[Order-Independent Indexability]
\label{order_independence_order}
Suppose submission costs $\{c_i\}$ are small enough that $u_i > u_j$ if and only if $u_i - \frac{c_i}{a_i} > u_j - \frac{c_j}{a_j}$, and that $a_i q_j$ is constant over all $i, j$. Then a strategy $\sigma^*$ is optimal if and only if it is monotone. 
\end{proposition}

The proof can be found in \blue{\ref{order_independence_order_proof}}. 
When there are more than two journals, this condition is rather stringent, since it requires that $a_i q_j$ is constant across \textit{all} indices, which implies that there can be at most \textit{two} unique tuples $(a_i, q_i)$. 
However, it is possible to extend Proposition \ref{order_independence_order} to allow for impose \textit{some} regularity on general problems. In particular, for any two boxes for which $a_i q_j = a_j q_i$ (that is, they \textit{locally} satisfy the order-independence condition), then it cannot be that $\sigma(t) = j$ and $\sigma(t + 1) = i$ for any $t$, where $u_i > u_j$ for any $t$. 
Nevertheless, this is still a knife-edge way to add submission costs. The next section deals with the general case. 

\subsection{Nonmonotonicity}
Suppose now journals did give feedback (i.e. $q_i > 0$), but there were no submission costs\footnote{This assumption is made for analytical simplicity and ease of exposition, as otherwise the joint condition on $u_i - \frac{c_i}{a_i}$ would be difficult to interpret.}. What is the value of submitting to a journal? First, there is the \textit{expected} value of acceptance to the journal, which ends the game and is given by $u_i a_i \mu(H)$; second is the \textit{informational} and \textit{feedback} value obtained in a rejection, which is governed by $q_i$. The value of a rejection to the agent can be completely summarized by the effect a rejection has on the agent's posterior belief about their type, since it affects their future expected payoff from submission to other journals. As a result, a box is more valuable \textit{conditional on a rejection} if $a_i$ is lower (since a rejection does not have as large of an effect the posterior belief of the high type) and if $q_i$ is larger (since there is a greater probability of state transition). 
When these effects are sorted in different directions, an interpretable index is unlikely, since the value of any box depends both on the history of submissions and its relative value to other boxes. 
As a result, in order to hope to come up with an interpretable index, I will make a stark assumption -- regularity -- which disciplines the value of boxes conditional on both \textit{acceptance} and \textit{rejection} in the same direction. This will be sufficient for a clean characterization when there are only two journals, and the last part of this section shows that a slight strengthening of regularity is minimally sufficient for the index on arbitrarily many boxes derived in Theorem \ref{weak_feedback}. 

\begin{definition}
\label{regularity}
    Boxes $\{u_i, a_i, q_i\}$ are (strictly) \und{regular} if, $u_i$ and $q_i$ are decreasing in $i$, and $a_i$ is increasing in $i$. 
\end{definition}

\begin{proposition}
\label{base_case}
      Suppose there are two boxes with labels $1, 2$ and outside option $u_\infty = 0$. If $\{u_i, a_i, q_i\}_{i = 1, 2}$ is (strictly) regular and the prior $\mu$ satisfies $\mu(H) \geq \frac{q_2}{q_2 + a_2}$, 
    then the monotone strategy is (uniquely) optimal. 
\end{proposition}
\begin{proof}
    Let $\mu$ be the prior. Since the outside option is $0$, the monotone strategy gives payoff 
    \[ u_1 a_1 \mu(H) + (1 - a_1 \mu(H)) u_2 a_2 f_1(\mu)(H) \]
   Hence, the monotone strategy is optimal if and only if 
    \[  u_1 a_1 \mu(H) + (1 - a_1 \mu(H)) u_2 a_2 f_1(\mu)(H) \geq  u_2 a_2 \mu(H) + (1 - a_2 \mu(H)) u_1 a_1 f_2(\mu)(H)\]
    There are two cases. First, suppose $u_1 a_1 \geq u_2 a_2$.
    We can then rearrange this problem into the condition that 
    \[ \frac{u_1 a_1}{u_2 a_2} \left(\mu(H) - (1 - a_2 \mu_t(H)) f_2(\mu)(H) \right) \geq (\mu(H) - (1 - a_1 \mu(H)) f_1(\mu)(H)) \]
    By our condition on the prior belief $\mu$, 
    \begin{align*}
    & \left(\mu(H) - (1 - a_2 \mu_t(H)) f_2(\mu)(H) \right) \geq 0 
    \\ \iff & \mu(H) \geq (1 - a_2 - q_2) \mu(H) + q_2 \iff \mu(H) \geq \frac{q_2}{a_2 + q_2}
\end{align*}
    so that we can bound the left hand side of the inequality above: we thus need only show that 
    \[(1 - a_1 \mu(H)) f_1(\mu)(H) \geq (1 - a_2 \mu_t(H)) f_2(\mu)(H) \iff \frac{f_1(\mu)(H)}{f_2(\mu)(H)} \geq \frac{(1 - a_2 \mu_t(H))}{(1 - a_1 \mu_t(H))}\]
   The proposition is then implied by the following computational lemma, whose proof is deferred to \blue{\ref{characterizing_regularity_proof}}. 
   \begin{lemma}
   \label{characterizing_regularity}
    Suppose $\{u_i, a_i, q_i\}$ is regular. Then for all $i < j$, 
   \[ \frac{f_i(\mu_t)(H)}{f_j(\mu_t)(H)} \geq \frac{1 - a_j\mu_t(H)}{1 - a_i \mu_t(H)} \]
for all interior prior beliefs $\mu_t(H) \in (0, 1)$. The converse holds under strict regularity. 
\end{lemma}
\noindent When $u_1 a_1 = u_2 a_2$, $a_1 = a_2$, so all inequalities hold exactly with equality instead. 
A corollary of this computation is that monotone strategies increase the total probability of exiting the game under weak feedback. Finally, suppose $u_2 a_2 > u_1 a_1$ instead. Our computation following 
Example \ref{strong_feedback} implies that the monotone strategy is optimal if and only if 
\[ (u_2 a_2 q_1 - u_1 a_1 q_2)(1 - \mu) \geq a_1 a_2(u_2 - u_1) \mu \]
which is vacuously true under regularity since the left hand side is nonnegative while the right hand side is nonpositive. 
\end{proof}

Why is the condition on the prior needed?
Recall from Example \ref{strong_feedback} that if the prior belief is too low and there is feedback, then the probability of state transition may eventually lead to an information structure where rejection makes a researcher \textit{more optimistic} about their paper, resulting in pathological (importantly, nonmonotone) behavior. Because, for any box characteristics $(a_i, q_i)$, there exist interior prior beliefs for which this may be true (in particular, when $\mu(L)$ is large), some condition on the prior is necessary to rule out this behavior. The requirement that $\mu(H) \geq \frac{q_2}{q_2 + a_2}$ is exactly this condition. 

What are the right ways to extend the argument above to arbitrarily many journals? First, it must be that, for \textit{any} belief which can be reached on-path, the agent is still sufficiently optimistic about their type that primitives do not display ``strong'' feedback.  Formally, 

\begin{definition}
    A prior $\mu$ satisfies \und{globally bounded weak feedback} given boxes $\{u_i, a_i, q_i\}$ if, for every strategy $\sigma$ and induced path of beliefs $\mu_s^\sigma$, $\mu_s^\sigma(H) \geq \frac{q_1}{q_1 + a_1}$. 
\end{definition}

Globally bounded feedback, along with regularity (Definition \ref{regularity}) is sufficient to govern \textit{local} behavior: by guaranteeing, at any history, the incoming prior belief satisfies regularity, the pairwise swap guarantees \textit{higher expected utility} among all strategies if the agent ignores continuation play past rejection by the latter of the two journals under consideration.
To understand the effect of \textit{local pairwise swaps} on continuation play, it is necessary to have control over both the change in the total probability of \textit{future exit} (which often decreases with monotone swaps) and on the entire future path of beliefs. 
These effects can be subtle and often work in opposite directions (Proposition \ref{minimal sufficiency} gives examples of this). In order to discipline these effects, then, I require that the acceptance payoffs of journals are \textit{strongly ranked}, in the sense that each journal is at least twice as good as every other lower-ranked journal. Together, these conditions will discipline continuation payoffs enough that local optimality implies monotonicity optimality. 

\begin{definition}
    Boxes $\{u_i, a_i, q_i\}$ are \und{exponentially regular} if they are regular and $u_i \geq 2 u_j$ for all $i \leq j$.
\end{definition}

\begin{theorem}[Monotone Optimality]
\label{weak_feedback}
    Suppose $\{u_i, a_i, q_i\}$ is \textit{exponentially regular} and $\mu$ satisfies globally bounded weak feedback under $\{u_i, a_i, q_i\}$. Then a strategy $\sigma^*$ is optimal if and only if it is monotone. 
\end{theorem}
\begin{proof}
The argument proceeds by induction. 
First suppose that $I = 2$. By our discussion preceding this proof, the definition of regularity, and Proposition \ref{base_case}, all (and only) monotone strategies are optimal.

Suppose now that there are $I + 1$ boxes which are exponentially regular, and the theorem is true for any collection of $I$ boxes which are exponentially regular. 
Note that exponential regularity is inherited by subsets, as is global bounded feedback for a fixed prior, since if it is satisfied on $\{u_i, a_i, q_i\}$ for prior $\mu_0$, then it is also satisfied on $\{u_i, a_i, q_i\}_{i \neq k}$ for prior $f_k(\mu_0)$. 

Fix any strategy $\sigma$ over all $I + 1$ boxes. We will find a sequence of permutations $\{\tau_k\}_{k = 1}^K$ such that applying the permutations at each step is weakly profitable, with a strict increase in payoffs for at least permutation. This will finish the inductive hypothesis and hence the proof. First, consider the permutation $\tau_1$ such that $\tau_1 \circ \tau_{\sigma}$ is monotone on the subgame with boxes $\{\tau_{\sigma}(s)\}_{s = 2}^{I + 1}$ (i.e. the subgame induced after rejection by the first journal that is submitted to under  $\sigma$). By the inductive hypothesis, the boxes available in this subgame are regular and satisfy global weak feedback with prior belief $f_{\tau_\sigma(1)}(\mu_0)$, and hence $\tau_1 \circ \tau_\sigma$ is weakly better than $\tau_\sigma$ (and strictly preferred if they are not equal). 

Suppose $\tau_\sigma(1) = k$. If $k = 1$, then we are done. Else, $(\tau_1 \circ \tau_\sigma)(2) = 1$. Consider the pairwise switch $\tau_2 = (1 \quad 2)$ which permutes the first and second boxes which are reached by $(\tau_1 \circ \tau_\sigma)$. 
We make a few observations. First, let $V_3^{\tilde \sigma}$ be the total value of the continuation game that occurs after being rejected from the first two journals under some strategy $\tilde \sigma$. Moreover, let $\sigma_2$ be a strategy which induces $\tau_2 \circ \tau_1 \circ \tau_\sigma$ and $\sigma_1$ be a strategy which induces $\tau_1 \circ \tau_\sigma$. Finally, let 
\begin{align*}
    p_{\tilde \sigma}(3) = \left(1 - a_{\tau_{\tilde \sigma}(3)} \mu_0(H) \right)\left(1 - a_{\tau_{\tilde \sigma}(2)} f_{\tau_{\tilde \sigma}(1)}(\mu)(H)\right)
\end{align*}
be the probability of rejection from the first two journals along the path induced by $\tilde \sigma$. Note 
\begin{align*}
    V_3^{\tilde \sigma} = \left( \sum_{s = 3}^{I + 1} p_\sigma(t) \left( u_{\tilde \tau_\sigma(s)} a_{\tilde \tau_\sigma(s)} \right) \mu_s^{\tilde \sigma} \right) + p_{\sigma(I + 2)} u_\infty 
    \\ \leq u_{\tau_{\tilde \sigma}(3)}\sum_{s = 3}^{I + 1}\left(  a_{\tau_{\tilde \sigma}(3)}  p_\sigma(t)  \mu_s^{\tilde \sigma} \right) \leq \bb{P}_3^{\tilde \sigma} u_{\tau_{\tilde \sigma}(3)} \leq u_{\tau_{\tilde \sigma}(3)}
\end{align*}
where we bound the payoff terms by regularity and the probabilistic terms by the law of total probability, since the probability of exit in the game conditional on reaching the third period must be weakly smaller than the probability of reaching the third period (here, we use the fact $p_\sigma(t) a_{\tilde \tau_\sigma(s)} \mu_s^{\tilde \sigma}$ is the probability of stopping at time $t$ under $\tilde \sigma$) for the two inequalities, respectively. From here, $\sigma_2$ is more profitable than $\sigma_1$ so long as 
\begin{align*}
    & u_1 a_1 \mu_0(H) + (1 - a_1\mu_0(H)) u_k a_k f_1(\mu_0)(H) + V_3^{\sigma_2}(f_k(f_1(\mu_0))) 
    \\ \geq & u_k a_k \mu_0(H) + (1 - a_k \mu_0(H)) u_1 a_1 f_k(\mu_0)(H) + V_3^{\sigma_1}(f_1(f_k(\mu_0))) 
\end{align*}
We can rearrange terms here so that the requirement is that 
\begin{align*}
    & u_1a_1 \left(\mu_0(H) - (1 - a_k \mu_0(H)) f_k(\mu_0)(H) \right) + V_3^{\sigma_2}
    \\ \geq & u_ka_k(\mu_0(H) - (1 - a_1 \mu_0(H)) f_1(\mu_0)(H))
    + V_3^{\sigma_1}
\end{align*}
Consider for a second the terms which are directly affected by the swap, i.e. those which do not involve the $V_3$ terms. Some algebraic rearranging gives the following chain of inequalities: 
\begin{align*}
     & u_1a_1 \left(\mu_0(H) - (1 - a_k \mu_0(H)) f_k(\mu_0)(H) \right) - u_ka_k(\mu_0(H) - (1 - a_1 \mu_0(H)) f_1(\mu_0)(H))
     \\ & = 
     u_1 \left( a_1\mu_0(H) - a_1(1 - a_k \mu_0(H)) f_k(\mu_0)(H) - \left[a_k\mu_0(H) - a_k(1 - a_1 \mu_0(H)) f_1(\mu_0)(H) \right] \right)
     \\ & + (u_1 - u_k)a_k \left( \mu_0(H) - (1 - a_1 \mu_0(H)) f_1(\mu_0)(H) \right)
     \\ & = u_1 \left( (1 - p_{\sigma_2}(3) ) - (1 -  p_{\sigma_1}(3)) \right) + (u_1 - u_k)a_k \left( \mu_0(H) - (1 - a_1 \mu_0(H)) f_1(\mu_0)(H) \right)
\end{align*}
The first term simplifies to $u_1( p_{\sigma_1}(3) - p_{\sigma_2}(3) )$, which is positive by Lemma \ref{order_indep_lemma} and regularity. The second term is positive because 
\[ \mu_0(H) \geq (1 - a_1 \mu_0(H)) f_1(\mu_0)(H) \iff \mu_0(H) \geq \frac{q_1}{a_1 + q_1} \]
which is true by the globally bounded feedback assumption. 
Consider now the remaining difference between terms which include the continuation probabilities. There are two cases.
First, if $V_3^{\sigma_1}\leq V_3^{\sigma_2}$, then the inequality clearly holds since
\[   u_1a_1 \left(\mu_0(H) - (1 - a_k \mu_0(H)) f_k(\mu_0)(H) \right) \geq u_ka_k(\mu_0(H) - (1 - a_1 \mu_0(H)) f_1(\mu_0)(H)) \]
is true from our discussion above and the continuation values only increase the gain from the pairwise switch. 
Now suppose instead that $V_3^{\sigma_1} > V_3^{\sigma_2}$. The difference between these two terms is given by 
\[ \sum_{s = 3}^{I + 2}  u_{\tau_{\sigma_1}(s)} a_{\tau_{\sigma_1}(s)} \left( p_{\sigma_1}(s) \mu_s^{\sigma_1} -  p_{\sigma_2}(s) \mu_s^{\sigma_2}\right) + u_\infty \left( p_{\sigma_1}(I + 2) - p_{\sigma_2}(I + 2)\right) \]
where we use the fact that $\sigma_1(s) = \sigma_2(s)$ whenever $s \geq 3$ by construction. 
Define the time 
\[ \xi = \min\left\{\inf \{s : p_{\sigma_1}(s) \mu_s^{\sigma_1} - p_{\sigma_2}(s) \mu_s^{\sigma_2} < 0 \}, I + 2 \right\}\]
to be the first time such that the total exit probability under $\sigma_2$ is greater than under $\sigma_1$. 
To understand the structure of $\xi$, we prove the following computational lemma which says that the values which define $\xi$ switch from being positive to negative exactly once. The proof is given in \blue{\ref{interval_structure_proof}.}
\begin{lemma}
\label{interval_structure}
    Suppose that 
    \[ p_{\sigma_1}(t) \mu_t^{\sigma_1} - p_{\sigma_2}(t) \mu_t^{\sigma_2} < 0\]
    Then for all $s > t$, 
    \[p_{\sigma_1}(s) \mu_s^{\sigma_1} - p_{\sigma_2}(s) \mu_s^{\sigma_2} < 0 \]
\end{lemma}
\noindent 
There are now a few cases. First, suppose $p_{\sigma_1}(I + 2) \leq p_{\sigma_2}(I + 2)$. 
We can write 
\begin{align*}
    V_3^{\sigma_1} - V_3^{\sigma_2} = & \sum_{s = 3}^{I + 1}  u_{\tau_{\sigma_1}(s)} a_{\tau_{\sigma_1}(s)} \left( p_{\sigma_1}(s) \mu_s^{\sigma_1} -  p_{\sigma_2}(s) \mu_s^{\sigma_2}\right) +  u_\infty  \left( p_{\sigma_1}(I + 2) - p_{\sigma_2}(I + 2)\right)
    \\ \leq & \sum_{s = 3}^{\xi - 1} u_{\tau_{\sigma_1}(s)} a_{\tau_{\sigma_1}(s)} \left( p_{\sigma_1}(s) \mu_s^{\sigma_1} -  p_{\sigma_2}(s) \mu_s^{\sigma_2}\right) 
    \\ \leq & u_3 \sum_{s = 3}^{I + 1} a_{\tau_{\sigma_1}(s)} \left( p_{\sigma_1}(s) \mu_s^{\sigma_1} -  p_{\sigma_2}(s) \mu_s^{\sigma_2}\right) + u_3 (p_{\sigma_1}(I + 2) - p_{\sigma_2}(I + 2))
    \\ - & u_3 \sum_{s = \xi}^{I + 1} a_{\tau_{\sigma_1}(s)} \left( p_{\sigma_1}(s) \mu_s^{\sigma_1} -  p_{\sigma_2}(s) \mu_s^{\sigma_2}\right) - u_3 (p_{\sigma_1}(I + 2) - p_{\sigma_2}(I + 2))
\end{align*}
The first equality follows from the fact $a_{\tau_{\sigma_1}(s)} = a_{\tau_{\sigma_2}(s)}$ for any $s \geq 3$ by definition. The first inequality follows by removing all negative terms (which occur at time $s \geq \xi$, by definition of $\xi$ and Lemma \ref{interval_structure}). The second inequality follows by bounding the middle term by pulling out $u_3 \geq u_{\tau_{\sigma_1(s)}}$ (noting that this gives an upper bound since all and only terms $s < \xi$ are positive) and then adding and subtracting the term on the last time. From here, note that  
\begin{align*}
    \sum_{s = \xi}^{I + 1} a_{\tau_{\sigma_1}(s)} \left( p_{\sigma_1}(s) \mu_s^{\sigma_1} -  p_{\sigma_2}(s) \mu_s^{\sigma_2}\right) + (p_{\sigma_1}(I + 2) - p_{\sigma_2}(I + 2)) 
\\ =  \left[ \sum_{s = \xi}^{I + 1} 
a_{\tau_{\sigma_1}(s)} p_{\sigma_1}(s) (s) \mu_s^{\sigma_1} + p_{\sigma_1}(I + 2) \right] - \left[ \sum_{s = \xi}^{I + 1} 
a_{\tau_{\sigma_2}(s)} p_{\sigma_2}(s) (s) \mu_s^{\sigma_1} + p_{\sigma_2}(I + 2) \right]
\end{align*}
The first term in square brackets is the probability that we exit at some time, $s$ under strategy $\sigma_1$ starting from time $s = \xi$, and the last term is the probability of not exiting. By the law of total probability, this is $p_{\sigma_1}(\xi)$. Similarly, the second term is exactly $p_{\sigma_2}(\xi)$. 
To understand the value of this difference, we need to bound $p_{\sigma_1}(\xi) - p_{\sigma_2}(\xi)$  from below; Lemma \ref{probability_bound} gives one such useful bound, which is proven in \blue{\ref{probability_bound_proof}}.

\begin{lemma}
\label{probability_bound}
Let $\xi$ be the stopping time defined above. Then 
    $|p_{\sigma_1}(\xi) - p_{\sigma_2}(\xi)| \geq p_{\sigma_2}(3) - p_{\sigma_1}(3)$. 
\end{lemma}
\noindent We can use Lemma \ref{probability_bound}, along with our previous computation, to obtain 
\[ V_3^{\sigma_1} - V_3^{\sigma_2} \leq u_3 \left(p_{\sigma_1}(3) - p_{\sigma_2}(3)\right) - u_3 \left(p_{\sigma_2}(3) - p_{\sigma_1}(3)\right) = 2 u_3 \left(p_{\sigma_1}(3) - p_{\sigma_2}(3)\right)  \]
which bounds the maximal possible change in the payoffs in continuation play given the pairwise swap. 
Finally, suppose instead that $p_{\sigma_1}(I + 2) - p_{\sigma_2}(I + 2) > 0$. Lemma \ref{probability_bound} implies we can rewrite this as 
\[  p_{\sigma_1}(I + 1) - p_{\sigma_2}(I + 1) - a_{\tau_{\sigma_1}(I + 1)}\left(p_{\sigma_1}(I + 1) \mu_{s - 1}^{\sigma_1} - p_{\sigma_2}(I + 1) \mu_{I + 1}^{\sigma_2}) \right) ) > 0 \]
Note $\mu_{I + 2}^{\sigma_1} \geq \mu_{I + 2}^{\sigma_2}$ by an application of the argument in Lemma \ref{interval_structure}. Thus, 
\begin{align*}
    p_{\sigma_1}(I + 2) - p_{\sigma_2}(I + 2) > 0 \iff 
 p_{\sigma_1}(I + 2)\mu_{I + 2}^{\sigma_1} - p_{\sigma_2}(I + 2)\mu_{I + 2}^{\sigma_2} > 0
\end{align*}
which, by the contrapositive of Lemma \ref{interval_structure}, implies $\xi = I + 2$. Thus,
\begin{align*}
    u_3 \left( \sum_{s = 3}^{I + 1} u_{\tau_{\sigma_1}(s)} a_{\tau_{\sigma_1}(s)} \left( p_{\sigma_1}(s) \mu_s^{\sigma_1} -  p_{\sigma_2}(s) \mu_s^{\sigma_2}\right) \right) + u_3 (p_{\sigma_1}(I + 2) - p_{\sigma_2}(I + 2)) 
   u_3 [p_{\sigma_1}(3) - p_{\sigma_2}(3)] 
\end{align*} 
is a valid upper bound. This exhausts all cases and implies  
\[ V_3^{\sigma_1} - V_3^{\sigma_2}  \leq 2u_3 [p_{\sigma_1}(3) - p_{\sigma_2}(3)] \] 
is a valid global bound on continuation play. $\sigma_2$ is more profitable than $\sigma_1$ whenever 
\[ u_1(p_{\sigma_1}(3) - p_{\sigma_2}(3)) \geq  2 u_3 \left(p_{\sigma_1}(3) - p_{\sigma_2}(3)\right) \geq V_3^{\sigma_1} - V_3^{\sigma_2} \]
which is true whenever $u_1 \geq 2 u_3$ (where $u_3 = u_{\sigma_2(3)}$, which can be either the second or third box), which is true by the exponential regularity assumption. 
Finally, to finish the proof, fix $\sigma_2$, and note $\tau_2 = (\tau_{\sigma_2} \circ \tau_{\sigma_1} \circ \tau_{\sigma})(1) = 1$. Applying the inductive hypothesis again to the boxes $\{\tau_2(s)\}_{s = 2}^{I + 1}$ to obtain a permutation $\tau_3$ such that $\tau_3 \circ \tau_2 = \text{id}$ and which which is profitable to $\tau_2$. But of course this is a monotone strategy, so only monotone strategies are optimal.
That all monotone strategies are optimal then follows from Lemma \ref{coherency} and finishes the proof. 
\end{proof}

Exponential regularity and globally bounded weak feedback are stark conditions, and there are many cases where they may not hold \textit{globally}. However, as in Proposition \ref{order_independence_order}, the pairwise swap approach in the proof imposes a few general constraints on optimality for \textit{any} sequence of journals. 
Moreover, despite their strength,\textit{regularly}\footnote{
\textit{Exponential} regularity need not be sufficient, and in fact numerical simulations suggest that smaller constants work for a wide range of variables. Finding the tightest constant on the rate of growth on payoffs is a subject for future work and beyond the current scope of the paper.} and globally bounded weak feedback are \textit{minimally sufficient}, in the sense that if any of the conditions are removed (treating each condition in the definition of regularity separately), there are examples in which monotone strategies are suboptimal while still satisfying all other conditions. 
Proposition \ref{minimal sufficiency} establishes both this claim and shows that in general, no prior-independent index may exist, which partially justifies both the globally bounded weak-feedback assumption and the focus on monotone strategies. The proof, which is done through a sequence of counterexamples, can be found in \blue{\ref{minimal_sufficiency_proof}}. 
Read together, these imply that absent further ad-hoc specifications on the \textit{magnitude} of effects, Theorem \ref{weak_feedback} (modulo the tightness on the speed of growth of the payoffs) is as general a characterization as possible in this framework. 

\begin{proposition}[Non-Indexability]
\label{minimal sufficiency}
    Fix any exponentially regular boxes $\{u_i, a_i, q_i\}$. There does not exist a prior-independent index. 
    Moreover, exponential regularity and globally bounded weak feedback are jointly minimally sufficient for monotone strategies to be optimal. 
\end{proposition}

\section{Conclusion}
\label{conclusion_section}

I considered a model of sequential search without recall where the stopping problem is exogenous, but where the value of future boxes depends on the history of past search. When intertemporal correlation happens only the \textit{informational} properties of search, the optimal policy is indexable and independent of the researcher's prior belief, and is qualitative similar to the baseline problem with endogenous stopping and independent boxes first studied by \cite{weitzman1979optimal}. 
When the intertemporal correlation can also lead to state transitions, indexability breaks down, and in general the optimal strategy may crucially depend on both the researcher's prior belief and the finer quantitative details of the search process itself. Despite this, I identify conditions under which, when feedback is sufficiently weak, the monotone strategies which are optimal when there is no feedback remain optimal. Methodologically, I make use of \textit{local pairwise switches} to establish this result, which establishes local necessary conditions on optimality even when my conditions -- exponential regularity and globally bounded weak feedback -- are not globally satisfied. 

I interpret these results in the context of a researcher deciding which journals to submit to. My results provide plausible conditions under which the researcher should submit to journals in order of their value whenever they are sufficiently differentiated. Importantly, the researcher's strategy is time-consistent: rejections should not change what their optimal strategy is, though it does affect their continuation value from continuing in the game. Moreover, my results also suggest that submission costs for journals with low submission rates may have large deterrent effects: they can exponentially decrease the value of search, and hence low submission costs may be sufficient to decrease submission congestion. 

To simplify the core economic friction, the model is highly stylized. Importantly, the payoffs to journals do not depend on their acceptance rate, and submission is conditioned on the paper being truly high quality (hence rejections are \textit{perfect bad news}). While I conjecture that the finer details of the information policy should not affect Theorem \ref{no_feedback}, their effect Theorem \ref{weak_feedback} is uncertain and should be a subject of future research. In addition, the type of feedback, modelled simply as a potential state transition, is also restrictive: future work should allow for richer information and state transition probabilities. Finally, relaxing the exogenous stopping and no-resubmission policies, while potentially implausible when studying journal-submission decisions, may be useful in understanding other sequential search settings such as the job market, and help better compare the current framework with multi-armed bandit models. The effects these modifications have on optimal search remain open and a fruitful area for future research. 

\newpage 
\bibliography{cites.bib}

\newpage 
\appendix
\singlespacing
\section*{Appendix}
\makeatletter\def\@currentlabel{Appendix}\makeatother
\label{Appendix}

\subsection*{I: Computational Lemmas.}
\makeatletter\def\@currentlabel{Appendix I}\makeatother
\label{Appendix I}

\subsubsection*{PROOF OF LEMMA \ref{order_indep_lemma}.}\label{order_indep_lemma_proof}
\begin{proof}
Let $i = 1$ and $j = 2$ without loss of generality by relabelling. Two applications of Bayes' rule imply
\[ f_2(f_1(\mu)) = \frac{(1 - a_2 - q_2)\frac{(1 - a_1 - q_1)\mu + q_1}{1 - a_1 \mu} + q_2}{1 - a_2\left(\frac{(1 - a_1 - q_1)\mu + q_1}{1 - a_1 \mu}\right)}\]
Multiplying the numerator and denominator by $1 - a_1 \mu$, we can simplify this to be 
\[ \frac{(1 - a_2 - q_2)(1 - a_1 - q_1)\mu + (1 - a_2 - q_2)q_1 + (1 - a_1\mu)q_2}{1 - a_1\mu - a_2(1 - a_1 - q_1)\mu - a_2 q_1}\]
If we set $\varphi = (1 - a_2 - q_2)(1 - a_1 - q_1)\mu + q_1 + q_2 - q_1 q_2$ and $\phi = 1 - a_1 \mu - a _2 \mu + a_1 a_2 \mu$ then we can write 
\[ f_2(f_1(\mu)) = \frac{\varphi - q_1 a_2 - q_2 a_1 \mu}{\phi - a_2 q_1 (1 - \mu)} \text{   and  } f_1(f_2(\mu)) =  \frac{\varphi - q_2 a_1 - q_1 a_2 \mu}{\phi - a_1 q_2(1 - \mu)} \]
We are interested in signing the difference between $f_1(f_2(\mu)) - f_2(f_1(\mu))$, which is given by 
\begin{align*}
   & \frac{\varphi - q_2 a_1 - q_1 a_2 \mu}{\phi - a_1 q_2(1 - \mu)}  - \frac{\varphi - q_1 a_2 - q_2 a_1 \mu}{\phi - a_2 q_1 (1 - \mu)}
    \\ = & \phi \varphi - \phi(a_1 q_2 + a_2 q_1 \mu) - \varphi(1 - \mu) a_2 q_1 + a_1 a_2 q_1 q_2(1 - \mu) + (a_2 q_1)^2\mu(1 - \mu) 
    \\ \quad - & \left[\phi \varphi - \phi(a_2 q_1 + a_1 q_2\mu) - \varphi(1 - \mu) a_1 q_2 + a_1 a_2 q_1 q_2 (1 - \mu) + (a_1 q_2)^2 \mu(1 - \mu)\right]
\end{align*}
Cancelling the common terms and factoring out $1 - \mu > 0$ implies the sign of this difference can be determined by signing
\[\mu \left((a_2 q_1)^2 - (a_1 q_2)^2 \right) + \phi(a_2 q_1 - a_1 q_2) - \varphi(a_2 q_1 - a_1 q_2) = (a_2 q_1 - a_1 q_2)\left[\mu (a_2 q_1 + a_1 q_2) - \left( \varphi - \phi \right)  \right]  \]
Some algebra implies that $\varphi - \phi$ can be rewritten into 
\[ (\mu - 1)(1 - q_1)(1 - q_2) + (a_2 q_1 + a_1 q_2)\mu\]
Since $(1 - q_1)(1 - q_2)$ is positive, the sign of $f_1(f_2(\mu)) - f_2(f_1(\mu))$ is determined by the sign of $a_2 q_1 - a_1 q_2$, exactly as desired. That equality holds if and only if $a_2q_1 = a_1 q_2$ then follows from the fact $\mu$ is full support, so $(\mu - 1)(1 - q_1)(1 - q_2) \in (0, 1)$.  
\end{proof}

\subsubsection*{PROOF OF LEMMA \ref{characterizing_regularity}.}\label{characterizing_regularity_proof}
\begin{proof}
Note that 
\[ 1 - a_j \mu(H) = 1 - \mu(H) + \mu(H) - a_j \mu(H) = (1 - a_j) \mu(H) + \mu(L) \]
which is exactly the total probability of rejection given belief $\mu(H)$. Thus, for $i \leq j$, 
\[ \frac{f_i(\mu)(H)}{f_j(\mu)(H)} = \frac{\frac{(1 - a_i)\mu(H) + q_i \mu(L)}{1 - a_i \mu(H)}}{\frac{(1 - a_j)\mu(H) + q_j \mu(L)}{1 - a_j \mu(H)}} = \frac{(1 - a_i)\mu(H) + q_i \mu(L)}{(1 - a_j)\mu(H) + q_j \mu(L)} \cdot \frac{1 - a_j \mu(H)}{1 - a_i \mu(H)} \]
Thus, we have that $\{u_i, a_i, q_i\}$ is regular if and only if 
\[ \frac{(1 - a_i)\mu(H) + q_i \mu(L)}{(1 - a_j)\mu(H) + q_j \mu(L)} \cdot \frac{1 - a_j \mu(H)}{1 - a_i \mu(H)} \geq \frac{1 - a_j \mu(H)}{1 - a_i \mu(H)} \iff  \frac{(1 - a_i)\mu(H) + q_i \mu(L)}{(1 - a_j)\mu(H) + q_j \mu(L)}  \geq 1\]
Rearranging gives that the requirement is that 
\[ (1 - a_i) \mu(H) + q_i \mu(L) \geq (1 - a_j) \mu(H) + q_j \mu(L) \iff (a_j - a_i) \mu(H) \geq (q_j - q_i) \mu(L) \]
As a result, regularity clearly implies the ratio condition. The converse. If $q_j > q_i$, then 
\[ \lims_{\mu(L) \to 1} (a_j - a_i)\mu(H) = 0 \geq (q_j - q_i) = \lims_{\mu(L) \to 1} (q_j - q_i) \mu(L) \]
is impossible. A similar argument, taking $\mu(H) \to 1$ instead, shows that $a_j < a_i$ is similarly impossible. This gives the converse under strict regularity. 
\end{proof}

\subsubsection*{PROOF OF LEMMA \ref{interval_structure}.}\label{interval_structure_proof}
\begin{proof}
First, we show that $\mu_s^{\sigma_1} > \mu_s^{\sigma_2}$ for all $s > 3$. This argument is by induction. From Lemma \ref{order_indep_lemma} and regularity, $\mu_3^{\sigma_1} > \mu_3^{\sigma_2}$, which gives the base case. Now fix any box $i$ and let $\mu > \mu'$. We have 
\begin{align*}
    & f_i(\mu) = \frac{(1 - a_i - q_i)\mu + q_i}{1 - a_i \mu} \geq \frac{(1 - a_i - q_i)\mu' + q_i}{1 - a_i \mu'} = f_j(\mu)
\\ \iff & (1 - a_i \mu')(1 - a_i - q_i)\mu + q_i(1 - a_i \mu') \geq (1 - a_i \mu)(1 - a_i - q_i)\mu' + q_i (1 - a_i \mu_i)
\\ \iff & (1 - a_i - q_i)(\mu - a_i \mu \mu' - (\mu' - a_i \mu \mu')) \geq a_i q_i(\mu' - \mu)
\\ \iff & (1 - a_i - q_i)(\mu - \mu') \geq a_i q_i (\mu' - \mu) \iff (1 - a_i - q_i) \geq -a_i q_i
\end{align*}
where the last equivalence uses the fact that $\mu > \mu'$ and divides through. This inequality must always hold, since it is equivalent to requiring that 
\[ 1 - a_i - q_i + a_i q_i \geq 0 \iff (1 - a_i)(1 - q_i) > 0 \]
which is true since $(a_i, q_i)$ are probabilities. This finishes the inductive step.
We can now get on with proving the lemma. Fix an $s \geq 3$. It is sufficient, by induction, to show the result for $t = s + 1$. Note 
\begin{align*}
    p_{\sigma_1}(s) \mu_t^{\sigma_1} - p_{\sigma_2}(s) \mu_t^{\sigma_2} < 0
    \implies \mu_s^{\sigma_2} \left(p_{\sigma_1}(s) - p_{\sigma_2}(s)\right) < 0 \implies p_{\sigma_1}(s) - p_{\sigma_2}(s) < 0
\end{align*}
We can compose $p_{\sigma_i}(s + 1)$ into a function of $p_{\sigma_i}(s)$ and $\mu_s^{\sigma_i}$ for any $i$ to obtain 
\begin{align*}
    p_{\sigma_1}(s + 1) - p_{\sigma_2}(s + 1) = & p_{\sigma_1}(s)(1 - a_{\tau_{\sigma_1}(s)} \mu_{s}^{\sigma_1}) - p_{\sigma_2}(s)(1 - a_{\tau_{\sigma_2}(s)} \mu_{s}^{\sigma_2})
    \\ & \leq \left(1 - a_{\tau_{\sigma_1}(s)} \mu_{s}^{\sigma_2}\right)\left(p_{\sigma_1}(s) - p_{\sigma_2}(s)\right) < 0.
\end{align*}
The first inequality uses the fact that $a_{\tau_{\sigma_1}(s)} = a_{\tau_{\sigma_2}(s)}$ and that $1 - \mu_s^{\sigma_2} > 1 - \mu_s^{\sigma_2}$ by the computation we did earlier in the proof of this lemma. This finishes the proof. 
\end{proof}

\subsubsection*{PROOF OF LEMMA \ref{probability_bound}.}\label{probability_bound_proof}
\begin{proof}
    We iteratively decompose $p_{\sigma_1}(s) - p_{\sigma_2}(s)$ for any $s$:
\begin{align*}
    p_{\sigma_1}(s) - p_{\sigma_2}(s) = p_{\sigma_1}(s - 1)(1 - a_{\tau_{\sigma_1}(s - 1)} \mu_{s - 1}^{\sigma_1}) - p_{\sigma_2}(s - 1)(1 - a_{\tau_{\sigma_2}(s - 1)} \mu_{s - 1}^{\sigma_2})
    \\ = p_{\sigma_1}(s - 1) - p_{\sigma_2}(s - 1) - a_{\tau_{\sigma_1}(s - 1)}\left(p_{\sigma_1}(s - 1) \mu_{s - 1}^{\sigma_1} - p_{\sigma_2}(s - 1) \mu_{s - 1}^{\sigma_2}) \right) 
    \\ = p_{\sigma_1}(s - 2) - p_{\sigma_2}(s - 2) - a_{\tau_{\sigma_1}(s - 2)}\left(p_{\sigma_1}(s - 2) \mu_{s - 2}^{\sigma_1} - p_{\sigma_2}(s - 2) \mu_{s - 2}^{\sigma_2})\right)
    \\ - a_{\tau_{\sigma_1}(s - 1)}\left(p_{\sigma_1}(s - 1) \mu_{s - 1}^{\sigma_1} - p_{\sigma_2}(s - 1) \mu_{s - 1}^{\sigma_2}) \right) 
    \\ \dots p_{\sigma_1}(3) - p_{\sigma_2}(3) - \sum_{r = 3}^{s - 1} a_{\tau_{\sigma_1}(r)}\left(p_{\sigma_1}(r) \mu_{r}^{\sigma_1} - p_{\sigma_2}(r) \mu_{r}^{\sigma_2}) \right) 
\end{align*}
by iterating over the terms. Set $s = \xi$. If the final term is nonnegative, then it is bounded by $|p_{\sigma_1}(3) - p_{\sigma_2}(3)|$ in absolute value. Otherwise, if the term is strictly negative, then because $p_{\sigma_1}(3) - p_{\sigma_2}(3)$ is positive, we have that 
\begin{align*}
    \left| p_{\sigma_1}(3) - p_{\sigma_2}(3) - \sum_{r = 3}^{s - 1} a_{\tau_{\sigma_1}(r)}\left(p_{\sigma_1}(r) \mu_{r}^{\sigma_1} - p_{\sigma_2}(r) \mu_{r}^{\sigma_2}) \right) \right| \leq  \left| - \sum_{r = 3}^{\xi - 1} a_{\tau_{\sigma_1}(r)}\left(p_{\sigma_1}(r) \mu_{r}^{\sigma_1} - p_{\sigma_2}(r) \mu_{r}^{\sigma_2}) \right) \right|
    \\ \leq \left| \sum_{r = 3}^{I + 1} a_{\tau_{\sigma_1}(r)}\left(p_{\sigma_1}(r) \mu_{r}^{\sigma_1} - p_{\sigma_2}(r) \mu_{r}^{\sigma_2}) \right) + p_{\sigma_1}(I + 1) - p_{\sigma_2}(I + 2) \right| = |p_{\sigma_1}(3) - p_{\sigma_2}(3)|
    \end{align*}
where we increase the (negative) absolute value by first removing positive terms (thus making the term more negative) and then adding further negative terms. This exhausts all cases, and as a result it must be that we have that 
\[|p_{\sigma_1}(\xi) - p_{\sigma_2}(\xi)| \geq p_{\sigma_2}(3) - p_{\sigma_1}(3) \]
because we know the right hand side is negative. This finishes the proof of the lemma. 
\end{proof}

\subsection*{II: Proof of Propositions \ref{order_independence_order} and \ref{minimal sufficiency}.}
\makeatletter\def\@currentlabel{Appendix II}\makeatother
\label{Appendix II}
\subsubsection*{PROOF OF PROPOSITION \ref{order_independence_order}.}\label{order_independence_order_proof}
\begin{proof}
Following the proof of Theorem \ref{no_feedback}, and because we satisfy the order-independence condition of Lemma \ref{order_indep_lemma}, it is sufficient to show that (1) pairwise swaps are optimal and (2) that the total probability of exit over the two nodes does not change. Fix boxes $i < j$. The payoff from submitting to $i$ and then $j$ is 
\[ u_i a_i \mu + (1 - a_i \mu) u_j a_j f_i(\mu) - c_i - (1 - a_i \mu) c_j = u_i a_i \mu + u_j a_j \mu - (a_i + q_i) u_j a_j \mu - u_j a_j q_i - c_i - c_j + a_i c_j \mu \]
Removing symmetric terms implies the difference between the payoff from submitting to $i$ and then $j$ (and vice versa) is 
\[ a_i a_j (u_i - u_j) \mu + q_i a_j (u_i - u_j) + a_i c_j \mu - a_j c_i \mu \]
where we use order independence to combining the $q_i a_j$ terms. Some algebraic rearranging implies this is 
\[ a_i a_j \mu\left( \left(u_i - \frac{c_i}{a_i} \right) - \left(u_j - \frac{c_j}{a_j} \right) \right) + q_i a_j (u_i - u_j) \]
which is positive by assumption.
Next we need to bound continuation play. The pairwise swap affects continuation play by potentially changing the belief after being rejected from $(i, j)$ (order independence guarantees this is not the case) and by affecting the total probability of exit along these notes. Note however that 
\begin{align*}
    (1 - a_i \mu)(1 - a_j f_i) = 1 - a_i \mu - a_j \mu + a_i a_j \mu + a_j q_i \mu - a_j q_i = 
    \\ 1 - a_i \mu - a_j \mu + a_i a_j \mu + a_i q_j \mu - a_i q_j = (1 - a_j \mu)(1 - a_i f_j)
\end{align*}
by the order-independence. Thus continuation play is not affected, finishing the proof. 
\end{proof}

\subsubsection*{PROOF OF PROPOSITION \ref{minimal sufficiency}.}\label{minimal_sufficiency_proof}
\begin{proof}
Suppose for simplicity there are only two journals. Example \ref{strong_feedback} implies the monotone strategy is optimal if and only if 
\[ (u_2 a_2 q_1 - u_1 a_1 q_2)(1 - \mu) \geq a_1 a_2 (u_2 - u_1)\mu \]
Moreover, the discussion following the example shows that, for prior beliefs $\mu(H) \in \left(\frac47, \frac{17}{29} \right)$, $u_1 > 2u_2$, $a_1 < a_2$, and weak-feedback is satisfied (since the relative ratios are $\frac{.2}{.2 + .2} = \frac12 < \frac{.4}{.3 + .4} = \frac47 < \mu(H)$), but for which the nonmonotone strategy is optimal. 
Consider the journals 

\begin{table}[h]
    \centering
        \begin{tabular}{|c|c|c|c|}
            \hline
                    & $u_i$ & $a_i$ & $q_i$ \\ \hline
            $i = 1$ &    2   &    0.8   &   0.4    \\ \hline
            $i = 2$ &    1   &   0.2  &    0.15   \\ \hline
        \end{tabular}
    \end{table}
\noindent The payoffs satisfy exponential regularity, and 
\[ \max\left\{\frac{q_1}{a_1 + q_1}, \frac{q_2}{a_2 + q_2} \right\} = \frac47\]
If $\mu = 0.9$, we are in the region where primitives display globally bounded weak feedback,
\[ (0.08 - 0.24)*0.9 < -(0.16)*.1 \iff - 0.144 < -0.016 \iff 0.144 > 0.016 \]
which is of course true. so the nonmonotone strategy is optimal. 
Finally, consider 
\begin{table}[h]
\centering
\begin{tabular}{|c|c|c|c|c|}
\hline
          & $u_i$ & $a_i$ & $q_i$ & $c_i$ \\ \hline
Journal 1 & 0     & 0.2  & 0.3   & 0     \\ \hline
Journal 2 & 1     & 0.3   & 0.2   & 0     \\ \hline
\end{tabular}
\end{table}
which violates the condition that $u_i$ is decreasing. Here, 
\[ \max\left\{\frac{q_1}{a_1 + q_1}, \frac{q_2}{a_2 + q_2} \right\} = \frac35\]
so at $\mu = 0.7$, the nonmonotone strategy is optimal since 
\[ .3*.3 = 0.09 > 0.042 = 0.06*.7\]
and applying the computation from Example \ref{strong_feedback}.
Next, journals satisfying exponential regularity but not globally bounded weak feedback. 
\begin{table}[h]
\centering
\begin{tabular}{|c|c|c|c|c|}
\hline
          & $u_i$ & $a_i$ & $q_i$ & $c_i$ \\ \hline
Journal 1 & 2     & 0.5  & 0.3   & 0     \\ \hline
Journal 2 & 1     & 0.6   & 0.2   & 0     \\ \hline
\end{tabular}
\end{table}
Clearly these journals satisfy exponential regularity. However, they fail globally bounded weak feedback when $\mu(H) = \frac{1}{20}$, since $\frac{1}{20} \leq \frac{0.3}{0.8}$. Moreover, since 
$- 0.18 * 0.95 < -0.2 * 0.05$, the nonmonotone strategy is optimal. Note this counterexample can be extended to any arbitrary number of journals, $I$ satisfying exponential regularity where these are the last two journals, using any prior belief $\mu$ such that following some submission order (for example, the monotone one) until the penultimate journal gives a belief $\mu_{I - 2}(H) < \frac{1}{16}$.
Finally, prior-independence. 
Example \ref{strong_feedback} shows that no prior-independent index exists, as optimality of each strategy depends crucially on whether $\mu(H) \geq \frac{17}{29}$. Moreover, for prior beliefs in $(\frac47, 1)$, the range of beliefs satisfies the conditions of Theorem \ref{weak_feedback}, but the optimal index still depends on whether or not $\mu(H) \geq \frac{17}{29}  > \frac47$. 
\end{proof}

\subsection*{III: The Normalization Lemma.}
\makeatletter\def\@currentlabel{Appendix III}\makeatother
\label{Appendix III}
The goal of this section is to formally state and prove the claim made in \blue{Section \ref{model_section}} that it is possible to normalize the outside option to $0$, so long $u_i \geq u_\infty$ for all $i \in \mcal I$. 

\begin{proposition}
\label{normalization}
    A strategy $\sigma$ is optimal for the boxes $\mcal B = \{u_i, a_i, c_i, q_i\}$ at outside option $u_\infty = 0$ if and only if $\sigma$ is still optimal for the boxes $\mcal B' = \{u_i - K, a_i, c_i, q_i\}$ for any $K \in \bb{R}$ with outside option $u_\infty = - K$. 
\end{proposition}
\begin{proof}
Let $\sigma$ be optimal on $\{u_i, a_i, c_i, q_i\}$.
First, note that under any strategies $\sigma$, the induced path of beliefs $\{\mu_t^\sigma\}$ is the same under both $\mcal B$ and $\mcal B'$ given the strategy $\sigma$ (as the belief evolution depends only on $f_i$, which depends only on $(a_i, q_i)$). 
We then have that $\sigma$ is optimal for $\mcal B'$ if and only if, for all other $\sigma'$, 
\begin{align*}
    & \sum_{t = 1}^{I}\left( \left( u_{\tau_\sigma(t)} - K \right) a_{\tau_\sigma(t)}\mu_t^\sigma(H) - c_{\tau_\sigma(s)} \right) p_\sigma(t) - K p_\sigma(I + 1) \geq
    \\ & \sum_{t = 1}^{I} \left( \left( u_{\tau_{\sigma'}(t)} - K \right) a_{\tau_{\sigma'}(t)} \mu_t^{\sigma'}(H) - c_{\tau_{\sigma'}(t)}\right) p_{\sigma'}(t) - K p_{\sigma'}(I + 1)
    \iff 
    \\ & \sum_{t = 1}^{I} \left(u_{\tau_\sigma(t)} a_{\tau_\sigma(t)}\mu_t^\sigma(H) - c_{\tau_\sigma(t)} \right) p_\sigma(t) - K \left(\sum_{t = 1}^{I} a_{\tau_\sigma(t)}\mu_t^\sigma(H) p_\sigma(t) + p_{\sigma}(I + 1) \right) \geq
    \\ & \sum_{t = 1}^{I} \left(u_{\tau_{\sigma'}(t)} a_{\tau_{\sigma'}(t)} \mu_t^{\sigma'} - c_{\tau_{\sigma'}(t)}(H)\right) p_{\sigma'}(t) - K \left(\sum_{t = 1}^{I} a_{\tau_{\sigma'}(t)}\mu_t^{\sigma'}(H) p_{\sigma'}(t) + p_{\sigma'}(t) \right)
\end{align*}
Note that for any strategy $\tilde \sigma$, 
\[ \left(\sum_{t = 1}^{I} a_{\tau_{\tilde \sigma}(t)}\mu_t^{\tilde \sigma}(H) p_{\tilde \sigma}(i) + p_{\tilde \sigma}(I + 1) \right) = 1\]
since this is a the probability measure induced over the terminal nodes by $\sigma$. 
Thus, we can cancel $-K$ from both sides and get that $\sigma$ is preferred to $\sigma'$ under normalized boxes $\mcal B'$ if and only if 
\[ \sum_{t = 1}^{I} \left( \left( u_{\tau_\sigma(t)} \right) a_{\tau_\sigma(t)}\mu_t^\sigma(H) \right) p_\sigma(t) \geq \sum_{t = 1}^{I} \left( \left( u_{\tau_{\sigma'}(t)} \right) a_{\tau_{\sigma'}(t)} \mu_t^{\sigma'}(H)\right) p_{\sigma'}(t) \]
which follows from the fact $\sigma$ is optimal under the boxes $\mcal B$ (and that $u_\infty$ under $\mcal B$ is 0). 
\end{proof}
Since the value of the outside option is taken for sure at the end of the game (there is no more room for rejection), it is without loss to let $u_\infty$ include any submission cost $c_\infty$ from ``submitting to the outside option.''

\end{document}